\documentclass[letterpaper,twocolumn,10pt, dvipsnames,svgnames]{article}
\usepackage{usenix}

\usepackage{tikz}

\usepackage{graphicx}
\usepackage{amsfonts}
\usepackage{appendix}
\usepackage{booktabs} 
\usepackage{tabularx}
\usepackage{makecell}
\usepackage{verbatim}
\usepackage{amsgen,amsmath,amstext,amsbsy,amsopn,amsthm}
\usepackage{cancel}
\usepackage{xspace}
\usepackage[most]{tcolorbox}
\usepackage{algorithm}
\usepackage{enumitem}
\usepackage[noend]{algpseudocode}
\usepackage{algorithmicx}
\usepackage{hyperref}
\usepackage{multicol}
\usepackage{multirow}
\usepackage{subcaption}
\usepackage{caption}
\usepackage{upquote}
\usepackage{soul}
\usepackage[normalem]{ulem}
\usepackage{cancel}
 \usepackage{bbm}
\usepackage{threeparttable}
\usepackage{float}
\usepackage{pifont}
\usepackage{fancyhdr}
\usepackage{mathtools}

\newsavebox{\measurebox}

\let\subparagraph\paragraph

\newcommand{\tsc}[1]{\textsc{#1}}

\newcommand{\ymark}{\text{\ding{51}}}
\newcommand{\YA}{$\ymark$}
\newcommand{\xmark}{\text{\ding{55}}}
\newcommand{\NO}{$\xmark$}

\makeatletter
\newcommand{\inmsg}[1]{%
(\tsc{#1}\checknextarg}
\newcommand{\checknextarg}{\@ifnextchar\bgroup{\gobblenextarg}{)~}}
\newcommand{\gobblenextarg}[1]{, #1\@ifnextchar\bgroup{\gobblenextarg}{)~}}
\makeatother

\newcommand{\bwp}{\theta_\mathtt{P}}
\newcommand{\prover}{\mathtt{P}}

\newtcolorbox[auto counter]{bbox}[2][]{%
    colback=white,
    colframe=black,
	colbacktitle=white!90!black,
    coltitle=black,
    fonttitle=\bfseries, 
    enhanced,
    attach boxed title to top left={yshift=-2mm, xshift=0.5cm},%
    #1,
}

\newtheorem{theorem}{Theorem}

\algnewcommand{\LeftComment}[1]{\Statex {\color{teal}\textbf{\(\triangleright\) #1}}} 
\algdef{SE}[AS]{As}{EndAs}{\textbf{as}\ }[1]{}%
\algdef{SxnE}[FOR]{ForInitialize}{EndForInitialize}[1]{\textbf{for}\ #1\ \textbf{initialize}:}
\algdef{SxnE}[IF]{Upon}{EndUpon}[1]{\textbf{upon receiving}\ #1\ \algorithmicdo}
\algdef{SxnE}[IF]{Init}{EndInit}{\textbf{Local variables initialization:}}

\renewcommand{\paragraph}[1]{\vspace{2pt} \noindent {\bf #1}}

\begin{document}

\date{}

\title{\Large \bf Proof of Backhaul: 
Trustfree Measurement of Broadband Bandwidth}

\author{
  {\rm Peiyao Sheng}$^{1}$, 
  {\rm Nikita Yadav}$^{1, 2}$, 
  {\rm Vishal Sevani}$^{1}$, 
  {\rm Arun Babu}$^{1}$, \\
  {\rm SVR Anand}$^{1}$, 
  {\rm Himanshu Tyagi}$^{1,2}$, 
  {\rm Pramod Viswanath}$^{1}$\\
  \\
   $^{1}$Kaleidoscope Blockchain Inc., 
   $^{2}$Indian Institute of Science
}

\maketitle

\begin{abstract}Recent years have seen the emergence of decentralized wireless networks consisting of nodes hosted by many individuals and small enterprises, reawakening the  decades-old dream of open networking. These networks have been deployed in an organic, distributed manner and are driven by new economic models resting on {\em tokenized}  incentives. A critical requirement for the incentives to scale is the ability to prove network performance in a decentralized ``trustfree" manner, i.e.,  a Byzantine fault tolerant  network telemetry system. 

In this paper, we present a Proof of Backhaul (PoB) protocol which measures the bandwidth of the (broadband) backhaul link of a wireless access point, termed prover, in a decentralized and trustfree manner. In particular, our proposed protocol is the first one to satisfy the following two properties: 
(1) {\it Trustfree.} Bandwidth measurement is secure against Byzantine attacks by collaborations of challenge servers and the prover.
(2) {\it Open.} The barrier-to-entry for being a challenge server is low; there is no requirement of having a low latency and high throughput path to the measured link. 
At a high-level, our protocol aggregates the challenge traffic from multiple challenge servers and uses cryptographic primitives to ensure that a subset of challengers or, even challengers and provers, cannot maliciously modify results in their favor. A formal security model allows us to establish guarantees of accurate bandwidth measurement as a function of the fraction of malicious actors. 

We implement our protocol with challengers spread across geographical locations. Our evaluation shows that our PoB protocol can verify backhaul bandwidth of upto 1000 Mbps with less than 8\% error using measurements lasting only 100 ms. The measurement accuracy is not affected in the presence of corrupted challengers. Importantly, the basic verification protocol lends itself to a minor modification that can measure available bandwidth even in the presence of cross-traffic.
 
Finally, the security guarantees of our PoB protocol output are naturally composable with ``commitments" on blockchain ledgers, which are commonly used for decentralized networks.

\end{abstract}
\section{Introduction}
\label{sec:introduction}

Decentralized networks have been in the making for decades. Starting with 
Software Defined Networking~\cite{kreutz2014software, kim2013improving} to simplify the hardware and open software~\cite{oran}
to facilitate application development, finally 
real-world deployments of decentralized Internet Service Providers (ISPs)~\cite{PMWANI} and decentralized Mobile Network Operators (MNOs)~\cite{Helium} have emerged. These decentralized networks have been made possible by the convergence of several engineering, business, and policy developments: the availability of cheap hardware for WiFi access points, and now even cellular base stations; the availability of 
cloud-native orchestration and AAA software~\cite{magma}; and the availability of lightly licensed spectrum for cellular communication~\cite{fcc-cbrs}. However, the real breakthrough in deployment comes with the emergence of a token-driven incentive ecosystem to bootstrap network growth and make individual hosts provide good network service. The leading exponent of such growth is the Helium network~\cite{Helium}, which is a multi-RAT network supported by hundreds of thousands of ``hotspots'' hosted by individuals.  

But a new engineering challenge has emerged -- we need to design {\em secure and decentralized} network telemetry. 
In centrally managed networks, network telemetry is used for performance measurement and subsequent optimization. In contrast, network telemetry plays a more pivotal role in decentralized networks. It is now needed to ensure that the network nodes provide the service that they are being paid for. 
For this purpose, there are two new requirements for decentralized network telemetry:
\begin{itemize}
    \item {\bf Trustfree.} The protocol is secure against Byzantine attacks by the parties involved.
    
    \item {\bf Open.} The barrier-to-entry for servers participating in decentralized telemetry is low. In particular, any node with a ``reasonably good"  internet connection should be able to participate.
\end{itemize}
 The measurements that we get as the output of such trustfree and open network telemetry protocols 
 can be viewed as a cryptographically secure proof of appropriate network performance. 
 
In this paper, we focus on measuring a specific network performance parameter which is of central importance in decentralized wireless network deployments. In such deployments, users are required to get a broadband connection with appropriate bandwidth, as a backhaul for the wireless access point. But how do we know that the user has indeed set up a good backhaul connection? Can we simply use any of the existing techniques from the large body of literature, spanning over decades, on bottleneck link-throughput measurement? It turns out none of the existing tools is applicable for our setting; below we point out shortcomings of prominent techniques and clarify our contributions.

\paragraph{Comparison with speedtest.}  Speedtest (\url{speedtest.net}) is a state-of-the-art bandwidth testing tool  widely used globally. Whenever a user (the ``prover") sends a measurement request, a nearby server  is  selected from a centralized challenger server pool. The selected server generates traffic continuously until the target link is saturated. This requires the challenger server to have a high bandwidth, low latency, low packet loss link to the prover; this  represents a high barrier to becoming a challenger. Furthermore, the measurements  rely on the rates of sending packets from challengers and acknowledgements from the prover -- an untrustworthy prover or challenger can adversely impact the measurement. Speedtest and similar architectures are unsuited for trustfree network telemetry.


\paragraph{Traffic aggregation.}  One way to allow more challengers to participate in the telemetry (and thus being more open)  is to aggregate traffic from multiple challengers. Such aggregation removes the requirement of high capacity for a single server to measure high-bandwidth links, by uniting a group of servers to generate sufficient traffic in {\em parallel}. While this technique can improve the accuracy  (e.g., recent works \cite{speedtest-multi,fastBTS,yang2021fast}), the method is not trustfree: a Byzantine prover can readily manipulate the measurement results with no check or balance.

\paragraph{Interactive Measurement.} To eliminate the need of trust on the prover, challengers should interact with other parties in the network to generate measurements. Popular interactive telemetry tools such as traceroute~\cite{jacobsonTraceroute} and pathchar~\cite{jacobsonPathchar} use the timing information obtained by combining the Internet control message protocol's (ICMP)  time-to-live (TTL) and  packet dropped messages to estimate link performance over the Internet. In particular, challengers estimate the round-trip time (RTT) to the two end-points of the link to be measured, the throughput is derived by dividing the packet size by the difference of RTT\footnote{The actual protocol is slightly different:  packets of different sizes are transmitted and  linear regression is used to aggregate the measurements, c.f.,~\cite{jacobsonPathchar, downey1999using}.}. Secure measurement, resistant to collusion between the prover and challengers, is not guaranteed in these protocols.

\paragraph{Our contributions.} We present the first protocol 
for measuring backhaul bandwidth that satisfies the aforementioned trustfree and open properties, c.f.,  \S\ref{sec:protocol}. 
Broadly, the protocol is built by implementing the following ideas:
\begin{enumerate}
    \item {\it Traffic aggregation.} We simultaneously send challenge traffic from multiple challengers to the prover. The duration of the challenge is chosen to be sufficiently high to ensure that traffic from all the challengers queues at the prover's backhaul link. 
    
    \item {\it Unforgeable probe.} The challengers are selected randomly from a larger pool and each sends digital signatures as traffic, so that no other party can forge the measurement probe. Furthermore, to limit the influence of any one challenger, we limit the amount of challenge traffic that can come from a single challenger. 
    
    \item {\it Short witness.} The prover can send a short message to the challengers to prove that it has received appropriate amount of data. As will be seen below, our security considerations require us to use a partially verifiable hash. For this purpose, we use a Merkle tree~\cite{merkle1987digital}.
    
    \item {\it Robust timing measurement.} We estimate the RTT for the overall challenge by taking median of the RTT measured by different challengers.
\end{enumerate}
We implement these steps and experimentally validate the design choices to identify the best performing configuration; see Figure~\ref{fig:protocol} for a depiction. The main contribution of this work is the trustfree property of the proposed protocol -- it is secure under a rigorous threat model that we outline in \S\ref{sec:model}. It is interesting to note that, even ignoring the security requirements, our proposed protocol is the first one that can measure bandwidth of hundreds of Mbps without requiring any specialized server with high throughput and low latency for challengers. We further extend this protocol to measure available bandwidth in the presence of cross-traffic, making it a truly distributed ``speedtest.''

\begin{figure}[h]
    \centering
    \includegraphics[width=0.45\textwidth]{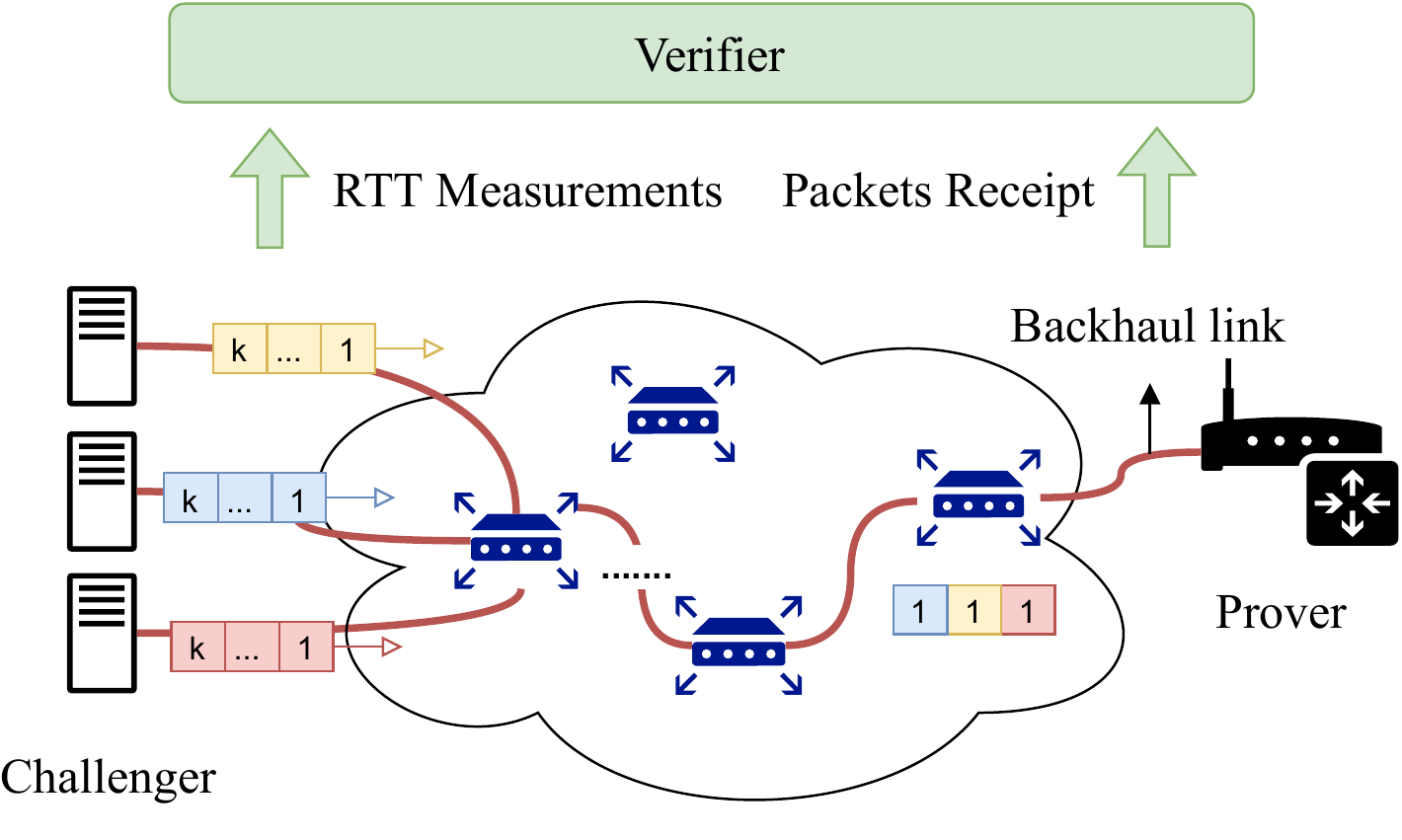}
    \caption{The multichallenger PoB Protocol.}
    \label{fig:protocol}
\end{figure}

    

\begin{figure*}[t]
\centering
\begin{tabular}{c}
    \begin{tabular}{cccc}
    \hline
    \textbf{Technique}          &  \textbf{Secure} & \textbf{Challenger BW $<$ Backhaul BW} & \textbf{Accuracy} \\
    \hline
    Pathchar \cite{downey1999using, jacobsonPathchar, mah2000pchar}                    &  \NO             &      \YA                             & Low      \\
    Packet dispersion based \cite{carter1996dynamic, carter1996measuring, lai2001nettimer, prasad2003bandwidth}           &  \NO             &      \NO                             &  \textbf{--}      \\
    Secure BW estimation \cite{karame2012security, snader2009eigenspeed, zhou2015magic}                &  \YA             &      \NO                             &  \textbf{--}      \\
    Multichallenger PoB                         &  \YA             &      \YA                             & High      \\
    \hline
    \end{tabular}
    \\
     (a) 
     \\
    \begin{tabular}{cccccc}
         \hline
         \textbf{Backhaul BW} & \textbf{Challenger BW} & \textbf{Challenge Data} & \textbf{Attack} & \textbf{Measured BW} & \textbf{Guaranteed BW} \\ 
         \textbf{(Mbps)} & \textbf{(Mbps)} & \textbf{(MB)} & & \textbf{(Error \%)} &  \textbf{(Mbps)} \\
         \hline
         250 & 25 & 3.44   & \textbf{--}  & 246 (1.6\%)  & 184 \\
         500 & 20 & 6.86   & \textbf{--}  & 475 (5\%)    & 356 \\
         750 & 75 & 10.31   & \textbf{--} & 705 (6\%)    & 529 \\
         1000 & 100 & 13.75 & \textbf{--} & 921 (8\%)    & 691  \\
         250 & 32 & 3.44  & Rushing       & 331 (0.6\%)  & 249  \\
         250 & 32 & 3.44  & Withholding  & 241 (3.6\%)   & 181  \\
         \hline 
    \end{tabular}
     \\
     (b)
\end{tabular}
\caption{(a) Comparison of our multichallenger PoB protocol with prior-art techniques. (b) Summary of our performance results with 10 challengers. We perform attacks with 2 corrupted challengers.}
\label{fig:result_summary}
\end{figure*}

We analyze the security of our multichallenger PoB protocol under a formal threat model
which allows any subset of parties (up to 1/3 challengers collaborating with the  prover) to maliciously deviate from the protocol. 
Since our probe is unforgeable, a corrupted prover still must get probe packets from the challengers. However, corrupted challengers, too, can modify the packet flow using two attacks: (i) the {\it withholding} attack where a corrupted challenger does not  send probe packets; and (ii) the {\it rushing} attack where a corrupted challenger coordinates with the corrupted prover to send the packets or their information quickly without using the challenged link. To compensate for the withholding attack, we must send more packets than the link bandwidth to have sufficient traffic even after withholding attack. To compensate for the rushing attack, we multiply the actual measured bandwidth with a correction factor to arrive at the {\it guaranteed bandwidth}. In addition, corrupted challengers can modify their outputs needed for verification. Specifically, they may report wrong RTT or they may claim modified packet data. We circumvent the former attack by taking a median of the measurements. To circumvent the latter attack, we use a Merkle tree which allows us to verify the consistency of the hash response from the prover with the data of uncorrupted challengers, without requiring correct data from the corrupted ones. 

Overall, denoting the fraction of corrupted challengers by $\beta$, we show that for $\beta<1/3$\footnote{We remark that the adversarial threshold can be $1/2$ if the verifier has access to a timer. 
See the discussion in \S\ref{sec:security}.}, our protocol does not allow any prover to inflate the bandwidth and allows
an honest prover to establish at least a fraction $(1-2\beta) / (1-\beta)$ of the true bandwidth.

\paragraph{ Implementation and evaluation.} 
To convert the idealized protocol into a practical tool, 
 we implement a variant of our protocol designed to address real-world issues (\S\ref{measurement-technique}) and thoroughly evaluate its performance (\S\ref{sec:experiment}). 
For instance of a real-world challenge,  measuring links with 100 Mbps and higher bandwidth (commonplace in broadband services)  requires  latency measurements with  an accuracy  that is hard to achieve   due to jitter in the Internet;   we elaborate on overcoming this challenge in  \S\ref{subsec:cross-traffic}. 
 
 In our evaluation, we focus on the loss of measurement accuracy when using multiple challengers and traffic aggregation; in particular, we consider
 the loss of accuracy due to:
 (i)  time synchronization errors and network jitter; (ii) computation time delays due to the use of digital signatures, hash computation, verification, and Merkle trees; and (iii) geographically spread challengers with heterogeneous capabilities. We also implement rushing and withholding attacks to illustrate that the security guarantees of the theory hold in practice. 
 Our main experimental results are summarized in
 Figure~\ref{fig:result_summary}. We report both the actual measured bandwidth and the output of our protocol -- the guaranteed bandwidth -- which is obtained by applying the correction factor $(1-2\beta) / (1-\beta)$.

\section{Background and Related Work}
\label{sec:related_work}

\paragraph{Bandwidth estimation.} The term bandwidth in the context of data networks quantifies the amount of data a network path can transfer per unit of time. Two metrics related to bandwidth are extensively investigated in the literature,
the maximum possible data rate called \textit{capacity} and the maximum available data rate called \textit{available  bandwidth}~\cite{prasad2003bandwidth}. Packet dispersion techniques~\cite{carter1996measuring, carter1996dynamic, salehin2013packet, keshav1991control, lai2001nettimer} are widely used to measure the capacity of the  bottleneck link in a network path. Some of the techniques to measure available bandwidth are outlined in~\cite{jain2002end, jain2002pathload, ribeiro2003pathchirp, strauss2003measurement, hu2003evaluation, banerjee2000estimating, melander2002regression, melander2000new,speedtest,fastBTS,yang2021fast,yang2022mobile}. Of these, tools such as Pathload~\cite{jain2002end, jain2002pathload} and Pathchirp~\cite{ribeiro2003pathchirp} create a short traffic load with different stream rates and observe the differences of one-way delay to adjust estimations. 
The state-of-the-art commercial tool Speedtest\cite{speedtest} employs a pool of servers with high bandwidth around the world to generate TCP traffic enough to saturate available bandwidth of the target link for a fixed duration. To improve accuracy, Speedtest and  recent work (e.g.,  FastBTS\cite{yang2021fast}) leverage concurrent connections to generate TCP traffic in parallel. Swiftest\cite{yang2022mobile} explores UDP to address limitations incurred by TCP-based methods such as slow start. Since bandwidth measurements play a critical role in optimizing centralized system performance and incentivizing decentralized services, other works shed light on the security of bandwidth measurements such as addressing inflation attacks in packet dispersion~\cite{karame2012security, zhou2015magic}. Secure bandwidth estimation tolerating malicious parties in peer-to-peer networks has been discussed in~\cite{snader2009eigenspeed}, where every participant in the network evaluates the bandwidth of others and the results from all parties are combined into one consensus vector using principle component analysis. This scheme only obtains opportunistic observations during normal operations, and any node with high bandwidth cannot get fully appraised since all the other nodes are constrained by their own bandwidth. In another direction,~\cite{ghosh2014torpath} proposes a proof system for network telemetry  for 
remunerating the relays in Tor network in proportion to the amount of data they transmit. The PoB proposed in this paper is aimed at measuring the backhaul bandwidth of end nodes in the Internet (e.g. WiFi access points and base stations). Further, we place no requirement on the bandwidth of the nodes measuring the backhaul; it can be much less than the backhaul bandwidth.

\paragraph{Per-hop capacity estimation.} Of all the bandwidth estimation techniques in literature,~\cite{jacobsonTraceroute, jacobsonPathchar, downey1999using, mah2000pchar, lai2000measuring, pasztor2002active, harfoush2003measuring} are closest to our work. These techniques can measure capacity for any link in an end-to-end path and so can be used to measure the prover backhaul, which is our goal. Traceroute~\cite{jacobsonTraceroute} and pathchar~\cite{jacobsonPathchar, downey1999using} make use of time-to-live (TTL) information in ICMP packets to control the packet drop at different intermediate hops to measure capacity of any link.~\cite{lai2000measuring, pasztor2002active, harfoush2003measuring} improve the approach used by pathchar~\cite{jacobsonPathchar, downey1999using} with variable packet sizes. However these techniques require precise timing measurements of the order of packet transmission times. For bandwidth in 100s of Mbps, the packet transmission times are of the order of tens of microseconds. Given the jitter in latency over the Internet, our experiments in \S\ref{sec:pchar} reveal that such precise timing measurements are difficult. Indeed,~\cite{harfoush2003measuring} reports errors over 20\%  for measuring bandwidths of 500Mbps or more.

\paragraph{Decentralized networks.} 
A common feature in every decentralized network deployment proposal  is a proof system that can be used to verify a particular network performance parameter. The participants are incentivized to help in this proof system and also stand to gain when they can establish their contribution to this parameter.
Helium~\cite{Helium} intends to unlock the potential of blockchains to establish a decentralized data network based on a tokenized incentive mechanism called proof-of-coverage. Hotspots are compensated for providing reliable coverage, to prove which challenge requests are issued regularly to random hotspots, 
who in turn are required to send beacons to other hotspots in the vicinity. Althea~\cite{Althea} aims to operate as a distributed ISP providing last-mile connectivity by creating a competitive platform and involving individual service providers into the market. Nodes maintain a route meter and accuracy score to assess the quality of neighbors to reach destinations and filter out inaccurate connections. To jointly address contractual and routing difficulties in inter-domain routing, Route Bazaar~\cite{route-bazaar} constructs a system to establish end-to-end connectivity agreements among mutually untrusted parties automatically. The performance of the path is guaranteed by periodically generated forwarding proofs recorded on blockchains, which contain information like encrypted path tags, traffic samples and timing and throughput measurements.

\section{The Multichallenger PoB Protocol}
\label{sec:protocol}
In this section, we formulate the PoB problem (\S\ref{sec:ps}), introduce main techniques (\S\ref{sec:protocol-overview}) and describe our multichallenger PoB protocol in details (\S\ref{sec:full_protocol}) . 

\subsection{Problem Statement}
\label{sec:ps}

We consider a system consisting of a group of end nodes such as base stations, WiFi access points and remote servers over the Internet willing to assist with backhaul bandwidth measurement.  All nodes are connected to the network core through one backhaul link,
simply referred to as backhaul from hereon, with an internal state $\theta$ representing the bandwidth of the link. We model the network core as a single point since fiber cables usually provide  extremely high bandwidth, e.g.,   100 Gbps. A PoB protocol allows a trusted verifier to use a subset of available nodes 
for securely measuring the backhaul bandwidth  
for a specific node called a {\em prover}, denoted $\prover$. 
The verifier can not observe the internal state $\theta$ of the prover directly. Instead, it needs to interact with the system by issuing a challenge request to the rest of the parties. To mitigate the dependence on (potentially corrupted) static nodes,  we require PoB protocols to randomly select $n$ participants to serve as {\em challengers}, denoted as $\{\mathtt{C}_1, \cdots, \mathtt{C}_n\}$. Suppose $f = \beta n$ challengers are corrupted;  $\beta$ represents the fraction of challengers in the overall pool that is adversarial. 


These  challengers are responsible for generating and sending probes to the prover and output the measurements to the verifier. The output of PoB protocol is an estimation of the backhaul bandwidth of the prover. It guarantees the following two security properties:
\begin{itemize}
    \item \textbf{Approximate completeness:} When the prover is uncorrupted, if the protocol outputs $\bwp'$, the actual bandwidth of the prover $\bwp$ satisfies $\bwp' \ge \alpha \bwp $ for a constant accuracy ratio $\alpha \in [0, 1]$.
    
    \item \textbf{Soundness:} The protocol will not output a bandwidth higher than $\bwp$, even when the prover is corrupted.
\end{itemize}

\paragraph{Other assumptions for theoretical analysis.} Our protocol makes use of digital signatures and collision resistant cryptographic hash functions. These primitives are assumed to be perfectly secure. 
We assume the network is synchronous and every challenger has access to a synchronized clock. Each node knows the public address and public key of others. We suppose there exists a trusted verifier such as a blockchain to broadcast information to the system.  
We remark that these assumptions are only made for our theoretical analysis; in the evaluation of our implementation we take into account the effect of deviation from these assumptions in practice.


\subsection{Protocol Overview} 
\label{sec:protocol-overview}

Heuristically, the protocol proceeds by randomly selecting a set of $n$ challengers from all the participants to send a train of probes to the prover (Figure~\ref{fig:protocol}). The protocol enforces packets from different challengers to arrive at the link to be measured around the same time. This traffic aggregation strategy effectively combines the group of challengers to an equivalent challenger with larger bandwidth and thereby
renders the prover's backhaul the bottleneck link.

Formally, suppose that the protocol starts at time $t_0$, and each challenger $\mathtt{C}_i$ starts to send a sequence of $k$ packets of size $b$ each at time $t_{i1}$, $1\leq i \leq n$.
We have the following two requirements: 
\begin{enumerate}
\item {\it Aggregation condition.} There is a $\theta_0\le \min(\theta_1, \cdots, \theta_n)$ such that the bandwidths $\theta_i$ of $\mathtt{C}_i$ satisfy 
    \begin{equation}\label{eq:queue}
        t_{0} + \frac{b}{\theta_0} = t_{11} + \frac{b}{\theta_1}  = \cdots = t_{n1} + \frac{b}{\theta_n}.
    \end{equation}
    \item {\it Bandwidth condition.}
    The quantity $\theta_0$ satisfies 
    \begin{equation}
    \label{eq:bw}
        (n-f)\cdot\theta_0 \ge \bwp.
    \end{equation}
\end{enumerate}
The ``aggregation condition" coordinates the arrival time of packets sent from various challengers, allowing multiple traffic flows to be effectively aggregated and merged into one stream at an appropriate rate. In this way, at least $(n-f)b$ bits of data are transmitted 
within the transmission time of one packet for a single challenger ($b / \theta_0$). Therefore, the equivalent bandwidth of the challenger group is enlarged by at least a factor of $(n-f)$. The ``bandwidth condition"  ensures that the prover's backhaul becomes the bottleneck link. 

While honest participants are supposed to correctly report their own bandwidth and send packets on time, corrupted parties can violate the conditions in  arbitrary ways. For instance, a corrupted challenger can {\em rush} the packets through extra links or refuse to send any packets.
We therefore require the prover to send back a response to all challengers on receiving $(n-f) k$ packets as a transmission receipt, since we can not expect more packets in the case of a {\it withholding} attack (detailed in \S\ref{sec:threat}). Then challengers measure the time it takes to transmit all these packets. Since corrupted challengers can claim arbitrary values, the median of all reported time is used to avoid manipulations and provide robust timing measurement.

\paragraph{Cryptographic primitives.} To save the bandwidth used for verification, the prover only sends back a short witness consisting of the hash of received packets to terminate the measurements. We define a hash function $Hash$ that takes any string as input and outputs a deterministic fixed-length random string. When the input is a set of messages, we assume the set will be serialized to a string to compute the hash. For verification, we ask each challenger to verify only packets sent by itself and employ the  Merkle tree construction to enable inclusion check with only partial data. A sequence of hashes can be aggregated using the function $MerkleRoot$ to a single cumulative hash. This technique reduces the verification overhead per challenger to $O(\log n)$. In addition, our protocol uses digital signatures to generate unforgeable probes and ensure traceability of bad behavior, for which the following functions are provided: a key generation function $keyGen$ which outputs a pair of secret and public keys, a signing function $sign(sk, msg)$ that allows anyone to sign an arbitrary message with a secret key $sk$, and a verification function $verify(pk, msg, \sigma)$ that checks whether the signature $\sigma$ is derived by signing given message $msg$ using the secret key paired with the public key $pk$.

\paragraph{Blockchain as a verifier.} Our PoB protocol is triggered by a challenge request issued from a verifier, who is also responsible for the broadcast of public parameters such as protocol start time $t_0$ and bandwidth requirement $\theta_0$. Generally, any trusted entity can play the role of a verifier. In tokenized decentralized settings, smart contracts supported by blockchains are a good fit to transparently generate, broadcast protocol parameters and coordinate measurement reports from multiple challengers. Implementing a version of our protocol with blockchain as a verifier  (deploying  appropriate smart contracts) is beyond the scope of this paper. 


\subsection{Full Protocol} 
\label{sec:full_protocol}
The full protocol contains two phases, a {\it measurement phase} described in Algorithm~\ref{alg:pob-measure}, where challengers generate and send packets, and a {\it verification phase} described in Algorithm~\ref{alg:pob-verify}, where the prover constructs proofs for the verifier. 
Finally, the verifier outputs the final results after verification. 

\begin{algorithm}[htbp]
\begin{algorithmic}[1]

\As{a challenger $\mathtt{C}_i$}

\State $t_0, m_0, \theta_0 \gets$ generated and broadcast by verifier

    \State measure its own bandwidth $\theta_i$ (require $\theta_i \ge \theta_0$). 
    
    \State generate $(pk_i, sk_i) \gets keyGen$, send $pk_i$ to verifier.

    \For{sequence number $q = 1, \cdots, k$}
    \State $t_{iq} \gets t_0 + q\cdot b / \theta_0 - b/\theta_i$ \label{alg:time}
        
    \State $\sigma_{iq} \gets sign(sk_i, (q, m_0)), m_{iq} \gets (i, q, \sigma_{iq})$
    
    \State send packet $m_{iq}$ to $P$ at $t_{iq} $
    
    \EndFor
    

\Upon{$(h_{1i}, h_2, \sigma_i)$ from $P$}
\If{$verify(pk_P, (h_{1i}, h_2),  \sigma_i)$ outputs $1$}
    \State record round trip time $\Delta_i \gets curTime_i - t_i$ 
\EndIf
\EndUpon
\EndAs{}

\As{a prover}

$\forall i \in [1,n], \mathcal{M}[i]\gets \emptyset$, $(pk_p, sk_p) \gets keyGen$

\Upon{packet $M'$ from $\mathtt{C}_i$} 
\State  $(i, q, \sigma) \gets M'$
\State add $(q, \sigma)$ to $\mathcal{M}[i]$ 
\If{$\sum_{i=1}^{n}|\mathcal{M}[i]| = (n-f)k $ }


    \State $\forall i\in [1, n], h_{1i} \gets Hash(\mathcal{M}[i])$
    \State $h_2 \gets MerkleRoot(\{h_{1i}\}_{i\in[1,n]})$
    \State $\sigma_i \gets sign(sk_p, (h_{1i}, h_2))$

    \State Send $(h_{1i}, h_2, \sigma_i)$ to $\mathtt{C}_i$ for all $i$.
\EndIf
\EndUpon
\EndAs{}
\end{algorithmic}
\caption{The Measurement Phase of PoB Protocol}
\label{alg:pob-measure}
\end{algorithm}

\begin{figure}[h]
    \centering
    \includegraphics[width=0.5\textwidth]{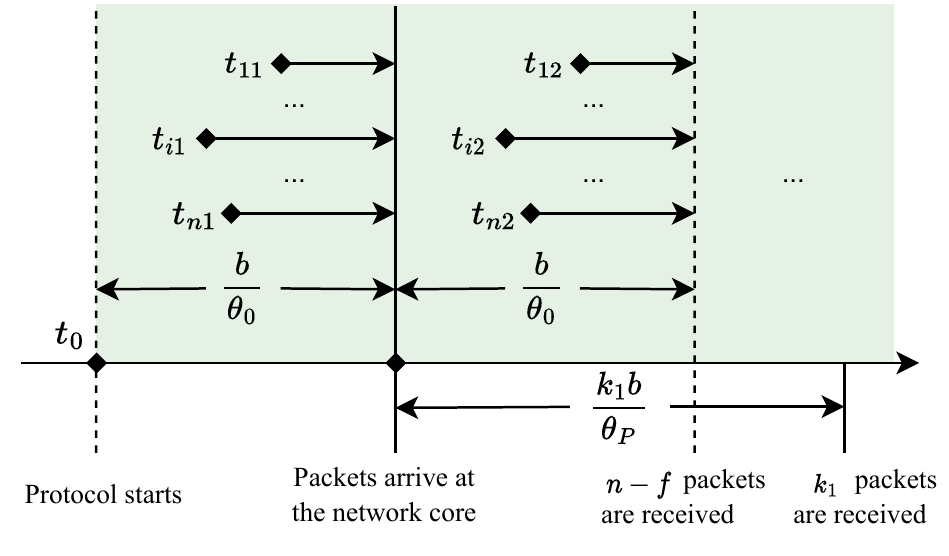}
    \caption{The measurement phase of PoB protocol. $t_{iq}$ is the time for challenger $\mathtt{C}_i$ to send the $q$-the packet. $k_1$ is the actual number of packets received with sequence number $1$.}
    \label{fig:timeline}
\end{figure}


\paragraph{Measurement phase.} At the beginning of the measurement phase, the verifier produces three public protocol parameters $(t_0, m_0, \theta_0)$ and broadcasts it to all challengers, where $t_0$ is the start time of the protocol, $m_0$ is a random message, $\theta_0$ is the global minimum bandwidth. Challengers are required to measure their own bandwidth $\theta_i$ and set up a key pair $(pk_i, sk_i)$, the public key is sent to the verifier. The time to start sending the first packet $t_{i1}$ is determined by Eq.~\eqref{eq:queue}. The challenger $\mathtt{C}_i$ generates a sequence of $k$ packets by signing the public message $m_0$ together with a sequence number $q$ and sends them one by one to the prover with a fixed duration $b / \theta_0$. The process is depicted in Figure~\ref{fig:timeline}, where the $q$-th packet of $\mathtt{C}_i$ is sent at time $t_{iq}$ (see line~\ref{alg:time} of Algorithm~\ref{alg:pob-measure}). On receiving the packets from the challengers, the prover separates the
messages from different challengers and adds them to corresponding sets.
When the total number of received packets reaches $(n-f)k$, the prover generates a response to broadcast to all challengers. This   terminates the measurements phase. The response contains (1) a receipt $h_{1i}$ to each challenger $\mathtt{C}_i$, which is the hash of all packets received from the same sender; and  (2) a Merkle root $h_2$ constructed from all receipts. All challengers record the round trip time $\Delta_i$ between the start time $t_{i1}$ and the time $curTime_{i}$ at which a valid response is received.

\paragraph{Verification phase.} 
In the verification phase, the prover is responsible for proving to the challengers the content of received packets. To that end, it constructs another response revealing the indices $B_i$ of packets received from $\mathtt{C}_i$
and showing the inclusion of each receipt in the Merkle tree with a Merkle proof $P_i$. It also sends the Merkle root to the verifier. On receiving the Merkle proof $P_i$ from the prover, the challengers reconstruct the receipt hash and the Merkle root. The 
challenger $\mathtt{C}_i$ forwards $\Delta_i$ and the number of packets sent by it to the verifier, after making sure that both the hashes are consistent. At the end of the second phase, the verifier aggregates these measurements from all the challengers 
about how long the measurement phase takes and how many packets are indeed received by the prover. It also forwards the reports from challengers to the prover, who checks the consistency and submits the packets and the Merkle proof in case disputes exist. Once the verifier has received ``sufficiently many" valid reports  
(specifically it waits to receive a confirmation from at least $n-f$ challengers with at least $(n-f)k$ packets in total), 
it calculates the final output bandwidth by dividing 
the number of total received packets 
by the median of the reported RTTs; see line~\ref{alg:output} of Algorithm~\ref{alg:pob-verify}.

\begin{algorithm}[t]
\begin{algorithmic}[1]

\As{a prover}
    \State $B_i \gets$ a bitmap of size $k$ where the $q$-th bit in the bitmap is set if some packet $(q, *) \in \mathcal{M}[i]$
    \State $P_i \gets$ the Merkle proof of $h_{1i}$
    \State send $ (B_i, P_i)$ to $C_i$, output \inmsg{Report}{$h_2$} to verifier.



\EndAs

\As{a challenger}
\Upon{$(B_i, P_i)$ from $P$}

\State Add all $(q, \sigma_{iq})$ to a set $ \mathcal{M}$ if the $q$-th bit is set to 1 in $B_i$. Check whether $Hash( \mathcal{M}) = h_{1i}$. 

\State Reconstruct the Merkle root $h'$ using $P_i$ and $h_{1i}$. Check whether $h_2 = h'$. 

\State If both two checks are passed, output \inmsg{Report}{$\prover, h_2, \Delta_i, | \mathcal{M}|$} to verifier. 

\EndUpon
\EndAs

\As{a verifier}

\State $ \mathcal{M}\gets \emptyset, cnt \gets 0$

\Upon{\inmsg{Report}{$h$} from prover $\prover$}

    \State record $h$

\EndUpon





\Upon{\inmsg{Report}{$\prover, h_2, \Delta_i, k_i$} from challenger $C_i$}

    \State Check $h_2 = h$, add $\Delta_i$ to $\mathcal{M}$, $cnt \gets cnt + k_i$


    \If {$cnt \ge (n-f)k $ and $|\mathcal{M}| \ge n-f$}

        \State $\Delta'\gets Median(\mathcal{M})$ 

        \State $\bwp' \gets \frac{cnt \cdot b \cdot (n-2f)}{\Delta' \cdot (n-f)}$ \label{alg:output}

        \State Output \inmsg{PoB}{$\prover, \bwp'$}
    
    \EndIf
\EndUpon

\EndAs

\end{algorithmic}
\caption{The Verification Phase of PoB Protocol}
\label{alg:pob-verify}
\end{algorithm}


\section{Security Model and Analysis}
\label{sec:model}
The primary challenge in trustfree networking is the inherent security vulnerability,  since any party can depart from the protocol at will and even collude with other parties to manipulate the results. In this section, we formalize a broad  threat model underlying measuring bandwidth, systematically examine the security issues to which the system is exposed (\S\ref{sec:threat}), and analyze the security guarantees for our protocol (\S\ref{sec:security}).

\subsection{Threat Model} 
\label{sec:threat}


We consider a static adversary allowed to corrupt at most $f$ among $n$ challengers before the protocol starts, the rest of uncorrupted challengers are referred to as honest. The prover can also be corrupted. In addition to the backhaul link indicated in the model in \S\ref{sec:ps}, we allow the adversary to access external communication channels. Specifically, the adversary has access to additional links with arbitrarily high bandwidth connecting to all the participants.
The corrupted parties can act arbitrarily in order to either inflate or deflate the measured bandwidth; we discuss prominent attacks below.

\paragraph{Withholding attack.} The measurement of bandwidth requires the challengers to send probes and measure the time that takes for the prover to receive the probes. The corrupted challengers who have been  bribed by the consumers or the competitors of a prover might be motivated to deflate the bandwidth estimation to reduce service costs. They can delay the sending of the packets to increase the observed RTT 
or even withhold the packets for the entire protocol. During the verification phase, corrupted challengers can also refuse to report verification results.
Moreover, the prover can also bribe the challengers to withhold packets during the measurement phase but report that the maximum number of packets have been sent in the verification phase.

\paragraph{Rushing attack.} Since a reasonable incentive system will
allow the participants to be compensated in proportion to their bandwidth, provers can collude with challengers to inflate bandwidth to get more rewards. During the measurement phase, instead of the backhaul link which is filled with the packets from uncorrupted challengers, corrupted challengers transmit packets through an extra link with an extremely high bandwidth to finish the measurements within a shorter time.  

\paragraph{Information sharing attack.} Besides the rushing attack, another way for the prover to get more information about the data from
the challengers than that was transmitted through the backhaul link 
is to exploit the information structure. In the verification phase, to facilitate the verification of whether the packets received by the prover are indeed those sent by the challengers, the challengers are required to provide the information related to packet generation. If the information to generate packets is much smaller than the actual packet data and is shared to the prover directly, the prover can also terminate the measurements much earlier since it can generate a fraction of packets by itself. 
For instance, in our protocol, corrupted challengers can send their secret keys to the prover.

\paragraph{Other attacks.} There are attacks that an {\em individual}  corrupted participant can launch. For example, any corrupted party can generate incorrect signatures, or send duplicate packets. Corrupted challengers can report wrong information during the verification phase. We note that these attacks are easy to detect with irrefutable evidence due to the use of signatures. 
Other common attacks include Denial-of-Service (DoS) attacks (especially if the protocol is implemented using the public IP address of the prover): 
specifically, a challenger with a high bandwidth link can flood the prover backhaul with invalid packets, preventing the valid packets sent by uncorrupted challengers from reaching the prover. 
Even a challenger who has not been selected for a particular challenge, but knows the time of the challenge, can disrupt the challenge similarly. To defend against DoS attacks, we  employ standard filtering techniques \cite{peng2003protection}. Disincentivizing such attacks via their detection (``slashing" in blockchain parlance \cite{slashing}) is the topic of future work. 




\subsection{Security Properties}
\label{sec:security}

\begin{theorem}[Soundness]
\label{thm:soundness}
 When $f < n / 3$, the prover cannot inflate the bandwidth.
\end{theorem}
\begin{proof}
According to the protocol, all packets with sequence number $q$ sent by uncorrupted challengers will arrive at the network core at $t_0 + q \cdot b / \theta_0$ and be added to the backhaul link queue $Q$. Because it takes at least $b / \theta_0$ to finish transmitting all packets with the same sequence number, the queue will never be empty during the measurement phase. Before sending the response, the prover waits for $K \ge (n-f)\cdot k$ packets, among which at most $fk$ packets come from corrupted challengers. These packets can be sent through an external link (rushing attack) or generated by prover directly if the secret keys are shared in collusion (information sharing attack). In either case they will not actually consume the bandwidth of the prover's backhaul link. Even so, there are still at least $K - fk \ge (n-2f)\cdot k$ packets sent by uncorrupted challengers. Since packets from uncorrupted challengers are not forgeable by anyone else, the earliest time at which the prover can send response is the time at which $(K - fk)$ packets from $Q$ get delivered, which is at least $t_R = t_0 + b / \theta_0 + (K - fk)b / \bwp$,
whereby the uncorrupted challengers will receive the response and at time $t_{R}^i > t_R$.

Since the verifier needs to collect at least $n-f$ time measurements, of which $n - 2f$ must be reported by uncorrupted challengers, the median $\Delta'$ of the RTTs
must be bounded by the minimum of honest measurements since $f<n/3$, in this way the estimated time will not get affected by individual misreports. Then, denoting the set of honest challengers as $H$, we have
$\Delta' \ge \min\{t_{R}^i\}_{i\in H}-t_0-b/\theta_0 > (K-fk)b/\bwp$ and  
\begin{align*}
    \bwp' &= \frac{K \cdot b \cdot (n-2f)}{\Delta' \cdot (n-f)} \le \frac{K  \cdot (n-2f)\cdot \bwp}{ (K - fk)  \cdot (n-f)}\le \bwp.
\end{align*}
\end{proof}

\begin{theorem}[Approximate completeness] When $f < n / 3$ and the prover is uncorrupted, the protocol will always output bandwidth with accuracy $\alpha = \bwp' / \bwp \ge (n-2f) / (n-f)$.

\end{theorem}

\begin{proof}
When the prover is uncorrupted, it waits for $(n-f)\cdot k$ packets to generate the response. Even under withholding attacks, where corrupted challengers never send their packets, $(n-f)\cdot k$ packets generated by honest challengers will arrive at the prover before  $t_R = t_0 + b/\theta_0 + (n-f)kb/\bwp$. Then all uncorrupted challengers receive the response at the same time and output to the verifier. Assuming that the size of the response is negligible, we have the median RTT $\Delta'= (n-f)kb/\bwp$. 
If corrupted challengers try to misreport the number of packets received by the prover, claim the proof sent by the prover is incorrect, or even withhold the measurement results, the prover can send the genuine packets it has received from the challenger to the verifier together with the Merkle proof. The verifier will reconstruct the Merkle root from the submitted partial data and Merkle proof to solve disputes.  Thus, even under attacks, the total number of packets 
are no less than $(n-f)k$. 
Consequently, the protocol will output 
\begin{align*}
    \bwp' = \frac{(n-f)k b \cdot (n-2f)}{\Delta' \cdot (n-f)} 
    \ge \frac{n-2f}{n-f}\bwp.
\end{align*}
\end{proof}


\paragraph{Remark.} {\bf (Adversarial threshold.)} Our protocol can tolerate up to a fraction $1/3$ of Byzantine challengers.
This threshold of $1/3$ arises from the requirement to ensure that a majority of $(n-f)$ RTT measurements are collected from uncorrupted challengers. This allows the verifier to terminate the collection responsively (or ``lazily'') when receiving enough reports without the requirement of a timer. However,
if the verifier has access to a timer with desired accuracy (roughly 100ms for us),
it can wait for a certain period (determined by maximal network delay and backhaul links transmission delays) to collect the measurements. In this case at least $n-f$ reports from honest challengers will be collected and the protocol is able to tolerate a fraction $1/2$ of Byzantine challengers.
\section{Protocol Implementation}
\label{measurement-technique}
In this section, we present the protocol implementation in a real system. 
Towards practicality, we discuss the factors that are not addressed in our theoretical modeling (\S\ref{subsec:measurement-technique}) which leads to the modifications in implementation to the basic form of the protocol (\S\ref{subsec:implementation}).


\subsection{Practical considerations}
\label{subsec:measurement-technique}

\paragraph{Challenger bandwidth.} In \S\ref{sec:full_protocol}, we assume each challenger can measure its spare bandwidth $\theta_i$ precisely. However, this bandwidth may be time-varying and
it will be difficult for the challenger to measure every time. 
We relax this requirement by allowing every challenger simply ensure that it has at least $\theta_0$ bandwidth available for the challenge. Here  
$\theta_0 = \bwp/(n-f)$ is the smallest value that satisfies the bandwidth condition in Eq.~\eqref{eq:bw}.
Each challenger will now send the challenge traffic at rate $\theta_0$.

\paragraph{Latency.} The key requirement of our technique is that the packets from each challenger reach the prover backhaul at the same time. The aggregation condition Eq.~\eqref{eq:queue} ensures this when there is no synchronization error or latency. However, in practice, a packet from the challenger $\mathtt{C}_i$ will take time $l_i$ to reach the prover, where $l_i$ is the one-way latency from challenger $\mathtt{C}_i$ to the prover. The value of $l_i$ can indeed vary for different challengers and to account for such varying latencies, we modify Eq.~\eqref{eq:queue} as
\begin{align*}
        t_{0} + \frac{b}{\theta_0} + l_0 = t_{11} + \frac{b}{\theta_0} + l_i  = \cdots = t_{n1} + \frac{b}{\theta_0} + l_n
\end{align*}
where $t_{i1}$ is the start time of challenger $\mathtt{C}_i$ to send the first packet. Note that $\theta_i$ is replaced by $\theta_0$ as in our implementation; challengers release packets at rate $\theta_0$. 

Likewise, the response packet from the prover will take time $l_i$ to reach challenger $\mathtt{C}_i$. Accordingly, $\Delta_i$ in Algorithm~\ref{alg:pob-measure} now changes to
   $ \Delta_i = curTime_i - t_i - 2 \cdot l_i$,
where $curTime_i$ is the time when challenger $\mathtt{C}_i$ receives the response from the prover.
For measuring $l_i$, before the challenge starts, each challenger sends 20 ICMP ping packets to the prover and takes the average across these 20 packets as $RTT$. We set the value $l_i$ as $RTT/2$.

\paragraph{Packet drops.} We have assumed that all the $k$ packets from a challenger will reach the prover. However, since all the challengers send the packets simultaneously to the prover, there will be buffer overflow at the last link of the prover and some packets will be dropped. We use UDP protocol for the challenge packets, so dropped packets will not be retransmitted. Since we use packet count as the termination condition, packet dropping will prevent
  the challenge from being terminated. In our experiments, we find that we can compensate for the packet drops by asking challengers to send $1.1k$ packets, i.e., assuming a packet drop rate of 10\%, this guarantees that the prover receives $(n-f)k$ packets and terminates.
  
\paragraph{Time synchronization.} We require that all the challengers are synchronized via Network Time Protocol (NTP)~\cite{ntp}. Note that NTP does not ensure perfect time synchronization, there can still be residual synchronization errors of the order of tens of milliseconds over the Internet~\cite{mills1989on}.

\paragraph{Computation overhead.} The use of cryptographic primitives like $Hash$ and $MerkleRoot$ (Algorithm~\ref{alg:pob-measure}) inevitably incurs computation overhead, which will delay the prover from sending responses to challengers and thereby add to an error in measurements. We detail empirical computation times of these primitives in \S\ref{subsection-perf-eval} as a function of the number of challengers and challenge duration for completeness.

 \subsection{Implementation}
 \label{subsec:implementation}
 We implement challengers and the prover as UDP socket applications in C++ and each challenger conducts measurements by sending UDP packets to the prover. Details are described below.

\paragraph{Digital signatures.}
As outlined in Algorithm~\ref{alg:pob-measure}, a challenger needs to sign each packet. We leverage the Edwards-curve digital signature scheme, Ed25519~\cite{josefsson2017edwards} for signature generation and verification as its computation overhead is low. 
Our measurement results indicate that if we use  challenge packets of 64 Bytes (the size of Ed25519 signature), measurement accuracy is affected especially if the challenger is connected over WiFi. For efficiency, modern generation WiFi uses packet aggregation where multiple packets from the network layer are combined into a single medium access control (MAC) layer packet of a larger size of up to 1 MB.~\cite{perahia2013next}. If we use smaller-sized 64 bytes UDP challenge packets, WiFi MAC aggregation is affected reducing the throughput i.e., $\theta_0$ for the challenger.  
To address this, we aggregate multiple signatures and send it as a single large packet. We use 1514 bytes challenge packets i.e., $b$ in Eq.~\eqref{eq:queue} is 1514 bytes (1472 byte UDP payload with a 42-byte header) which contain 23 different 64 byte signatures. We use the OpenSSL based implementation of Ed25519~\cite{openSSL_ed25519}. 

\paragraph{Hashing and verification.} As described in Algorithm~\ref{alg:pob-measure}, upon receiving the required number of total packets, the prover generates a hash for each challenger $\mathtt{C}_i$, i.e., $h_{1i} \gets Hash(M[i])$ where $M[i]$ is the set of all the signatures received from challenger $\mathtt{C}_i$. The prover then generates a $MerkleRoot$ of all the hashes from all the challengers. For generating the hash we use $sha256$ hash function via the implementation~\cite{sha256} and for generating the Merkle root, we use the following C++ open source implementation~\cite{merkletree}. The prover sends $h_{1i}$ and $MerkleRoot$ as response to the challenger $\mathtt{C}_i$. The response packet is a UDP packet with a payload of 64 bytes as it contains two 256-bit hashes. In the verification phase (Algorithm~\ref{alg:pob-verify}), the prover sends bitmap $B_i$  and Merkle proof $P_i$ to challenger $\mathtt{C}_i$, who then verifies the Merkle proof and sends RTT $\Delta_i$ and number of its packets received by the prover, to the verifier.

\paragraph{Precomputing the signatures.}  Signature generation  incurs computation time too and our benchmarking of OpenSSL implementation~\cite{openSSL_ed25519} of Ed25519 indicates that generating one signature of 64 bytes takes about 50-60 microseconds ($\mu s$) on a resource-constrained Linux system consisting of 1 GB of RAM and 1 CPU core. For each packet, a challenger has to generate 23 signatures which will incur a maximum time of $23*60 \approx 1.4 ms$. As the signature generation time is more than the packet transmission time of about $1.2ms$ even at $\theta_0 = 10Mbps$, in our implementation challengers precompute all the signatures before the challenge begins. This can be done after the challenger receives the challenge request and while measuring the ping latency $l_i$.

\paragraph{Benchmarking the technique.} Thus, making use of multiple challengers with the additional requirement of security introduces more sources of errors. Particularly, $l_i$ is not a constant and has some jitter. NTP synchronization can result in error of tens of milliseconds over the Internet. Computation overhead of hash and Merkle tree generation adds delay. Given these sources of error, we evaluate the accuracy as a function of challenger duration and the number of challengers. 

\section{Experimental Evaluation}
\label{sec:experiment}
In this section, we offer performance evaluations of our implementation and highlight how existing per-hop capacity estimation techniques fail to give accurate results for backhauls of 100 Mbps or more (\S\ref{subsection-perf-eval}). We stress test our experiments  under Byzantine attacks to evaluate the security of the protocol (\S\ref{subsection-security-eval}).


\paragraph{Experimental Setup.} Our setup consists of a diverse set of challengers in terms of computation capability and geographical location. We carry out experiments with a maximum of ten challengers. The details of the prover and different challengers are listed in Table~\ref{tbl:exp-setup}. Challengers 1-3 are connected to the Internet via WiFi links, while other challengers have a wired link.

\begin{table}[htb]
\centering
\begin{tabular}
{c|rr|c|c}
\hline
\textbf{} & \multicolumn{2}{c|}{\textbf{Compute Parameters}} & \textbf{Location} & \textbf{RTT} \\
\textbf{} & \textbf{RAM (GB)} & \textbf{CPU} & \textbf{} & \textbf{(ms)} \\
\hline  
Prover & 1 & 1 & AWS X &   \\
Ch. 1-3 & 12-16 & 4-8 & Y & 25  \\
Ch. 4-5 & 1 & 1 & AWS Z & 198  \\
Ch. 6-10 & 1 & 1 & AWS X & 1  \\
\hline
\end{tabular}
\caption{Experimental setup details for the prover and challengers. Location of the nodes are in different continents and are anonymized.}
\label{tbl:exp-setup}
\end{table}




\subsection{Performance Evaluation} 
\label{subsection-perf-eval}

First, we benchmark the accuracy of our measurements by  carrying out experiments without corrupted challengers. We evaluate the performance in the presence of corrupted challengers later in \S\ref{subsection-security-eval}.


\paragraph{Measurement accuracy.}
To study how accuracy varies with challenge duration and the number of challengers, we conduct experiments by adjusting the number of selected challengers from 4 to 10 and challenge duration from 25 ms to 200 ms for a prover with backhaul bandwidth ($\bwp$) of 250 Mbps. 

Challenge duration is the time required to transmit the required amount of packets i.e., $(n-f)k$ packets through the prover backhaul. Individual challengers will take longer to complete the challenge due to their latency, $l_i$, and the fact that they send $1.1k$ packets to account for packet drops.
We rate-limit the prover backhaul to 250 Mbps using the Linux utility tc~\cite{tc-linux-manual} and set the bandwidth of each challenger ($\theta_0$) as $\bwp/n$, where $n$ is the number of challengers.

\begin{figure}[htb]
    \centering
    \includegraphics[width=0.48\textwidth,trim={2.0cm 4.2cm 2.1cm 4.7cm},clip]{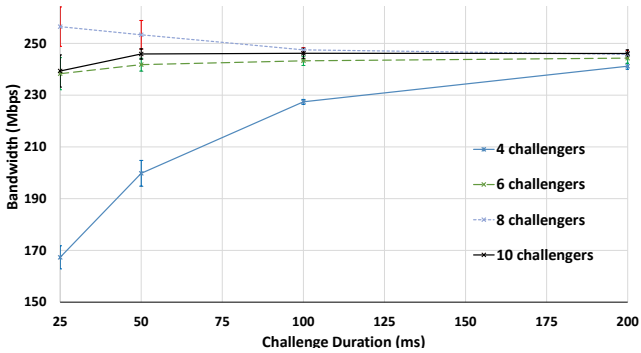}
    \caption{Backhaul measured by our technique for different challenge durations. Error bars $\approx$ std. deviation.}
    \label{fig:numofChallengers}
\end{figure}

Fig.~\ref{fig:numofChallengers} shows the backhaul measured by our technique for varying number of challengers and challenge durations as 25ms, 50ms, 100ms and 200ms. For each challenge duration, ten experiments are carried out. We plot the average and standard deviation for ten experiments in Fig.~\ref{fig:numofChallengers}.

As can be seen from Fig.~\ref{fig:numofChallengers}, with the number of challengers set to 4, the measured backhaul is only about 167 Mbps for 25ms challenge duration, but when the challenge duration is increased to 200 ms, the measured backhaul increases to about 241 Mbps with an error of about 4\%. On the other hand, when the number of challengers is increased to 6 or more, the measured backhaul has an error of less than 5\%, even for 25 ms challenge. However, the standard deviation for 25 ms and 50 ms experiments is higher. The measurement accuracy increases and the standard deviation decreases, if the challenge duration is increased to 100 ms or more. For 100ms, we observe an error of less than 5\% for six or more challengers. 

\begin{figure}[htb]
    \centering
    \includegraphics[width=0.48\textwidth,trim={2.0cm 4.8cm 2.0cm 4.7cm},clip]{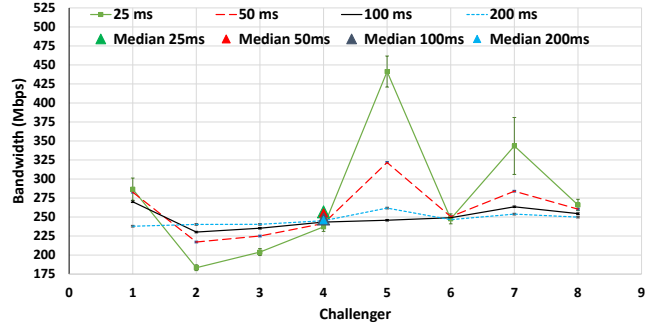}
    \caption{Backhaul measured by each challenger ($n=8$) for different challenge durations. Error bars $\approx$ std. deviation. }
    \label{fig:exp1-duration}
\end{figure}

We look at how measurement accuracy is affected as a function of the challenge duration for the case of $n=8$. Fig.~\ref{fig:exp1-duration} shows the backhaul measured by each of eight challengers, as an average across ten experiments for different challenge durations. As can be seen from Fig.~\ref{fig:exp1-duration}, the backhaul measured by individual challenger shows higher error when challenge duration is 25ms or 50ms. For example, the backhaul measured by challenger 5 is about 440 Mbps and 320 Mbps for 25 ms and 50 ms duration. However, the error decreases when challenge duration is increased to 100 ms or more. Note that standard deviation across 10 experiments for each challenger also decreases as the challenger duration is increased.

\paragraph{Sources of error.} Interestingly it can be seen from Fig.~\ref{fig:exp1-duration} that some challengers measurement of prover backhaul that is higher than the actual value of 250 Mbps. This is due to errors affecting measurement accuracy (see \S\ref{measurement-technique}) such as time synchronization and jitters in latency. We observe that due to these errors, there is a time difference of 20-30 ms between the first packet from the first and the last challenger reaching the prover backhaul. To compensate for the packet drop rate of 10\% (see \S\ref{measurement-technique}), each challenger sends 10\% more data. So, the challengers that start late might receive the response from the prover before they finish sending their share of challenge packets, if sufficient challenge packets have been received by the prover from the challengers that start early. Such late starting challenger's backhaul estimate may be higher than the actual value. However, the median evaluation at the final step, which is primarily designed for security, also provides robustness against such outliers. Consequently, our measurement accuracy increases as we increase the number of challengers. 

The computation overhead of hash and Merkle tree generation also adds to the measurement error. We observe that the computation overhead for the case of 4 challengers for 25 ms challenge duration is about 500$\mu s$, while for 10 challengers for 200 ms challenge duration is about 3ms.

\paragraph{Amount of data.} For the PoB protocol designed to handle the packet drop rate of 10\%, the total amount of data required for different challenge duration for prover backhaul of 250 Mbps is given in Table~\ref{tbl:total-data}. 

\begin{table}[htb]
\centering
\begin{tabular}
{|c|c|c|c|c|}
\hline
\textbf{Expt. Duration} & 25ms & 50ms & 100ms & 200ms \\
\hline
\textbf{Data (MB)} & 0.86 & 1.71 & 3.44 & 6.88 \\
\hline
\end{tabular}
\caption{Amount of challenge data required.}
\label{tbl:total-data}
\end{table}

As seen from Fig.~\ref{fig:exp1-duration}, for challenges with 100-ms duration we get a good accuracy for each challenger. Thus, our results show that our technique can measure 250 Mbps backhaul in 100 ms with about 3.5 MB of data and an error of less than 5\%, when 6 or more uncorrupted challengers are involved.

\paragraph{Comparison with a single challenger.} With a single challenger that has a bandwidth of 250 Mbps, we could measure prover backhaul of 250 Mbps with less than 2\% error with challenge duration being only 10 ms and the amount of data required is about 345 KB. 
Multichallenger technique requires larger challenge duration due to the aforementioned  errors. As the duration of the challenge is longer, the amount of data used correspondingly increases. But the primary benefit of multichallenger technique is that each challenger requires much smaller bandwidth. With ten challengers, each challenger requires a bandwidth of only 25 Mbps to measure prover backhaul of 250 Mbps.

\paragraph{Accuracy for larger prover backhauls.} The experimental results of accuracy for larger prover backhauls (500 Mbps to 1000 Gbps) with 10 challengers are tabulated in Table~\ref{tbl:larger-backhaul}. We observe that the measurement error grows as prover backhaul increases; however even for prover backhaul of 1000 Mbps, the measurement error is less than 8\%.


\begin{table}[htb]
    \centering
    \begin{tabular}{|c|c|c|c|}
        \hline
        \textbf{Backhaul (Mbps)} & 500 & 750 & 1000  \\
        \hline
        \textbf{Measured BW (Mbps)} & 474.7 & 705.4 & 921.4  \\
        \hline
        \textbf{Error (\%)} & 5.1 & 5.9 & 7.9  \\
        \hline
    \end{tabular}
    \caption{Measured bandwidth for larger prover backhauls.}
    \label{tbl:larger-backhaul}
\end{table}


\begin{table}[htb]
    \centering
    \begin{tabular}{|c|c|c|c|}
        \hline
        \textbf{Backhaul (Mbps)} & 500 & 750 & 1000  \\
        \hline
        \textbf{Overhead (ms)} & 4.6 & 7.3 & 10.2 \\
        \hline
    \end{tabular}
    \caption{Computation overhead}
    \label{tbl:overheads}
\end{table}

One reason for higher measurement error as prover backhaul increases is the increasing computation time for hash and Merkle tree construction. Table~\ref{tbl:overheads} shows the computation overhead for various prover backhauls. The computation overhead for 1000 Mbps is about 10 ms which is 10\% of the challenge duration of 100 ms. These experiments suggest that as the prover backhaul increases, the computation overhead increases. So for even larger prover backhaul than 1000 Mbps, the challenge duration should be increased.

\label{subsec:cross-traffic}
\paragraph{Effect of cross traffic.} Our PoB protocol terminates when $(n-f)k$ packets are received by the prover. The number of packets $k$ sent by each challenger is determined by the prover backhaul and challenge duration. However, if there is cross-traffic, the available bandwidth at the prover will be reduced and the challenge packets may experience a larger drop rate than 10\% that we assume for our experiments. In this situation the experiment may not terminate.

We propose a modification to measure the available bandwidth in the presence of cross-traffic, 
up to a fixed accuracy $\delta$. The protocol repeats the basic PoB protocol, but instead
of verifying $\bwp$, it verifies iteratively $\bwp^{(1)}=\delta, \bwp^{(2)}=2\delta, ..., \bwp^{(\ell)}=\ell \delta$ and so on till $\bwp$. In more detail, we proceed as follows.
\begin{enumerate}
    \item At step $i$, execute multichallenger PoB protocol with $\bwp^{(i)} =i \delta$. 
    Note that each challenger must release challenge traffic at rate $\theta_0= \bwp^{(i)}/n$ at this step. 
    \item Each challenger sets a timeout of $5\times$ (challenge duration). If the response from the prover is not received during this period, the challenger declares {\tt not terminate}. 
    If majority of the challengers declare {\tt not terminate}, we say that the protocol does not terminate.
    \item If the protocol for the $i$th step terminates, increment $i\leftarrow i+1$ and repeat the steps above.
    \item Else if the protocol for the $i$th step does not terminate, output the bandwidth obtained in the previous execution of the PoB protocol.
\end{enumerate}
Using the approach above, we carried out experiments to measure available bandwidth in the presence of different amounts of cross-traffic. The backhaul of the prover is set to 250 Mbps and the number of challengers is 10.

\begin{table}[htb]
    \centering
    \begin{tabular}{|c|c|c|c|}
    \hline
     \textbf{Available BW (Mbps)} & 220 & 140 & 90  \\
    \hline
     \textbf{Measured BW (Mbps)} & 219.65 & 144.63 & 104.15  \\
    \hline
    \end{tabular}
    \caption{Measured bandwidth in the presence of cross traffic.}
    \label{tbl:cross-traffic}
\end{table}

Table~\ref{tbl:cross-traffic} summarizes the results. $\theta_{P,0}$ is set to 40 Mbps and $\delta$ to 20 Mbps. The measured bandwidths are close to available bandwidths, except for 90 Mbps. Probably for the case of available bandwidth as 90 Mbps, the presence of challenge traffic alters the available bandwidth. 

\label{sec:pchar}
\paragraph{Comparison with pathchar.} As we outlined earlier, the performance of pathchar depends on how robust are minimum delay 
estimates over the Internet and how long will it take for us to get a robust estimate. 
Thus, to evaluate the performance of pathchar,
we measure the RTT to the prover node using ping for 15 different packet sizes in multiples of 100 Bytes, starting from 100 Bytes and ending at 1500 Bytes. We ran the experiment five times and took 500 measurements for each packet size. Pathchar~\cite{jacobsonPathchar, downey1999using} suggests taking the minimum RTT for each packet size and fitting the linear least squares line to the data. 

\begin{figure}[htb]
    \centering
    \includegraphics[width=0.45\textwidth]{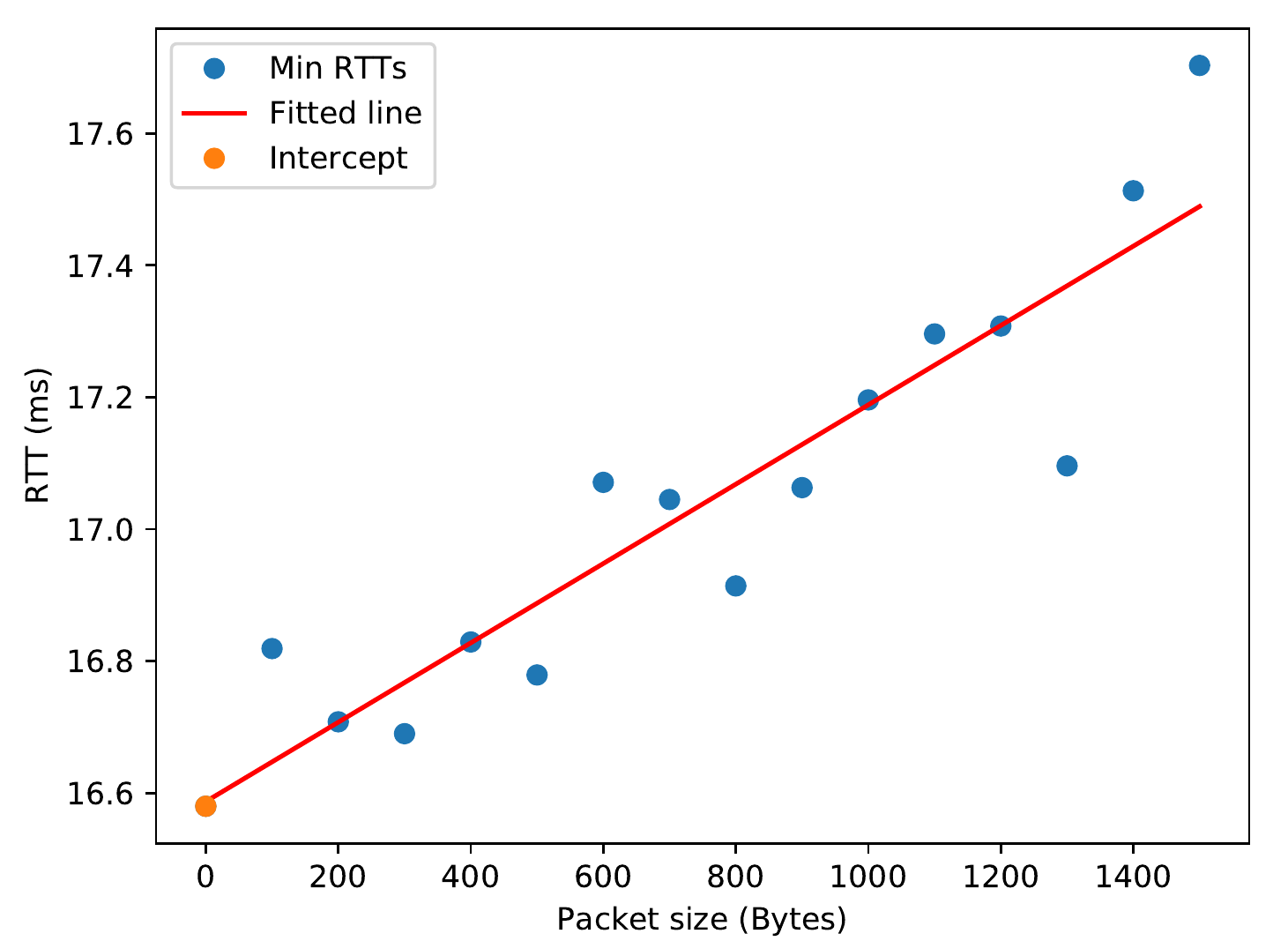}
    \caption{Minimum round trip time to the prover versus packet size in Bytes. The line shows the linear least squares fit.}
    \label{fig:ping-measurement}
\end{figure}

\begin{table}[htb]
    \centering
    \begin{tabular}{|c|c|c|c|c|c|}
    \hline
    \textbf{Expt.} & 1 & 2 & 3 & 4 & 5 \\
    \hline
    \textbf{RTT} & 16.589 & 16.511 & 16.635 & 16.565 & 16.500 \\
    \hline
    \end{tabular}
    \caption{Linear least squares line's intercept values for each experiment's dataset. RTT is in $ms$.}
    \label{tab:ping-analysis-data}
\end{table}

Figure \ref{fig:ping-measurement} shows the minimum RTT (ms) versus packet size in bytes and the fitted line for the first experiment. Table \ref{tab:ping-analysis-data} shows the y-intercepts for five experiments; note that the y-intercept represents latency. We can see that the y-intercepts have a difference of 50-100 ms. Thus, we can say that the jitters experienced over the internet is not negligible;
in particular, we cannot estimate the minimum latency  below accurancy of 50-100 microseconds. Consequently, it is not feasible to use pathchar to measure 100 Mbps or  higher bandwidth.






\subsection{Security Evaluation}
\label{subsection-security-eval}
We carry out experiments to study how measurement results are effected in the presence of malicious challengers. We choose total number of challengers as $n=10$ and malicious challengers $f=2$. We carry out measurements for two different prover backhauls of 100 Mbps and 250 Mbps, with a challenge duration of 100 ms.

\begin{figure}[htb]
    \centering
    \includegraphics[width=0.48\textwidth,trim={2.1cm 4.9cm 2.1cm 4.9cm},clip]{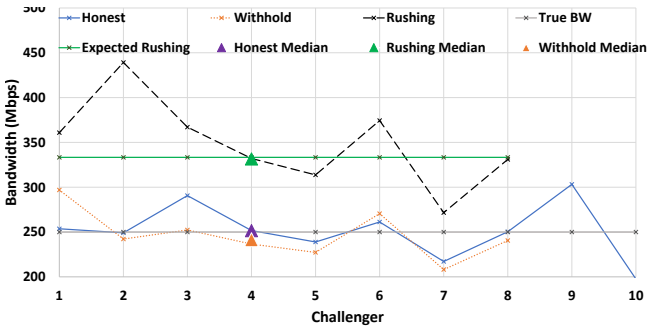}
    \caption{Backhaul measured by each challenger in case of withholding and rushing attack. Prover backhaul is 250 Mbps.}
    \label{fig:security-250Mbps}
\end{figure}

Fig.~\ref{fig:security-250Mbps} shows the average bandwidth measured by each challenger across ten experiment runs for the case when prover backhaul is 250 Mbps. As can be seen from Fig.~\ref{fig:security-250Mbps}, the measured backhaul in the case when all challengers are honest (honest median in Fig.~\ref{fig:security-250Mbps}) is 251.5 Mbps, while the measured backhaul in the case of withholding attack (Withhold Median in Fig.~\ref{fig:security-250Mbps}) is 241.4 Mbps. The expected measured backhaul in case of withholding attack is 250 Mbps. So the measured accuracy in case of withholding attack is within 4\%. Note that this is the accuracy of the measurement technique. Our PoB protocol will apply a correction factor $\alpha = (n-2f)/(n-f)$ (Algorithm~\ref{alg:pob-verify}) and output the guaranteed bandwidth for the prover as $241.4\alpha \approx 181$ Mbps which is about 28\% less than the prover backhaul of 250 Mbps. 

In case of a rushing attack, the measured bandwidth is inflated to $250 / \alpha \approx 333$ Mbps. The measured backhaul (rushing median in Fig.~\ref{fig:security-250Mbps}) 331.5 Mbps matches our theoretical prediction. The output of the protocol in this case will be $331.5\alpha \approx 249$ Mbps, which is only 1 Mbps less than the prover backhaul.

Note that our security guarantees require us to curtail bandwidth inflation. Indeed, 
we can observe that even under a rushing attack, the guaranteed bandwidth does not exceed the actual bandwidth. This is enabled by multiplying by a shrinkage factor to compensate for adversarial
challengers trying to help the prover to claim an inflated bandwidth.
However, this comes at the cost of lower guaranteed bandwidth even when all challengers are reporting honestly. 
After repeating the experiment for a backhaul of 100 Mbps, the results stay similar, validating  our theoretical predictions.






\section{Conclusion and Discussion}
\label{sec:discussion}
\paragraph{Summary}. Trustfree telemetry is a central problem in decentralized networks. Our  Proof of Backhaul protocol addresses a core requirement by providing a secure and accurate backhaul bandwidth measurement service for wireless access points while also allowing open participation. The protocol is operated by a group of challengers, whose latency and bandwidth can be ordinary, with the goal of measuring a prover hotspot who may have a high-bandwidth backhaul link. We have established a trust model for the PoB problem, designed precise specifications of the PoB protocol, and tested a high-performance, low-overhead implementation.

\paragraph{Improving Accuracy}. Our protocol guarantees soundness and completeness of backhaul measurements with a reasonable accuracy in the presence of Byzantine parties. The accuracy ratio $(1-2\beta)/(1-\beta)$ is determined by the Byzantine fraction due to a correction made for an unavoidable rushing attack -- corrupted challengers can always rush their packets through an external high-bandwidth link to lower RTT and inflate backhaul bandwidth to be measured. However, such backdoor links may incur substantial costs in practice, necessitating a more relaxed threat model and a family of extended protocols. Without rushing links, we equip PoB protocols with a shuffle phase where a pair of challengers are asked to jointly sign packets 
. This mechanism improves accuracy by making information sharing attack harder in a probabilistic manner, with a cost of higher communication overhead for verification. Designing a secure and efficient proof structure for such a shuffle protocol is an active area of research. 

\paragraph{Cross-traffic.} 
In our proposed method for handling cross-traffic in \S\ref{subsec:cross-traffic}, we run experiments for increasing values of bandwidth below the claimed link capacity. This requires fresh data to be sent for each value and increases the amount of data needed. To reduce the data consumption, a naive approach could be that the prover replies to the challengers with the number of packets received in
a fixed duration.
An alternative approach is to send intermediate responses when appropriate amounts of data are received. Both approaches cannot guarantee a fixed accuracy for different available bandwidths. Designing a protocol which is more data efficient and has such guarantees is an open problem.

\bibliographystyle{plain}
\bibliography{sample}

\begin{thebibliography}{10}

\bibitem{merkletree}
\url{https://github.com/IAIK/merkle-tree}.
\newblock [Online; accessed 18-September-2022].

\bibitem{sha256}
C++ sha256 implementation.
\newblock \url{http://www.zedwood.com/article/cpp-sha256-function}.
\newblock [Online; accessed 18-September-2022].

\bibitem{fastBTS}
Fast internet speed test.
\newblock \url{https://fast.com/}.
\newblock [Online; accessed 13-October-2022].

\bibitem{speedtest-multi}
Multi-server testing.
\newblock
  \url{https://www.ookla.com/articles/how-ookla-ensures-accurate-reliable-data-2021}.
\newblock [Online; accessed 11-October-2022].

\bibitem{ntp}
Ntp: The network time protocol.
\newblock \url{http://www.ntp.org/}.
\newblock [Online; accessed 18-September-2022].

\bibitem{openSSL_ed25519}
Openssl ed25519 implementation.
\newblock \url{https://www.openssl.org/docs/man1.1.1/man7/Ed25519.html}.
\newblock [Online; accessed 18-September-2022].

\bibitem{PMWANI}
{PM-WANI Central Registry}.
\newblock \url{https://pmwani.gov.in/wani}.
\newblock [Online; accessed 18-September-2022].

\bibitem{speedtest}
Speedtest.
\newblock \url{https://www.speedtest.net}.
\newblock [Online; accessed 13-October-2022].

\bibitem{tc-linux-manual}
Tc - traffic control, linux manual.
\newblock \url{https://man7.org/linux/man-pages/man8/tc.8.html}.
\newblock [Online; accessed 18-September-2022].

\bibitem{banerjee2000estimating}
Suman Banerjee and Ashok~K. Agrawala.
\newblock Estimating available capacity of a network connection.
\newblock In {\em Proceedings IEEE International Conference on Networks 2000
  (ICON 2000). Networking Trends and Challenges in the New Millennium}. IEEE,
  2000.

\bibitem{slashing}
Vitalik Buterin.
\newblock Minimal slashing conditions.
\newblock
  \url{https://medium.com/@VitalikButerin/minimal-slashing-conditions-20f0b500fc6c}.
\newblock [Online; accessed 18-September-2022].

\bibitem{carter1996dynamic}
Robert~L. Carter and Mark~E. Crovella.
\newblock Dynamic server selection using bandwidth probing in wide-area
  networks.
\newblock Technical report, Boston University Computer Science Department,
  1996.

\bibitem{carter1996measuring}
Robert~L Carter and Mark~E Crovella.
\newblock Measuring bottleneck link speed in packet-switched networks.
\newblock {\em Performance evaluation}, 27:297--318, 1996.

\bibitem{route-bazaar}
Ignacio Castro, Aurojit Panda, Barath Raghavan, Scott Shenker, and Sergey
  Gorinsky.
\newblock {Route Bazaar: Automatic Interdomain Contract Negotiation}.
\newblock In {\em 15th Workshop on Hot Topics in Operating Systems (HotOS XV)},
  Kartause Ittingen, Switzerland, May 2015. USENIX Association.

\bibitem{downey1999using}
Allen~B Downey.
\newblock Using pathchar to estimate internet link characteristics.
\newblock {\em ACM SIGCOMM Computer Communication Review}, 29(4):241--250,
  1999.

\bibitem{fcc-cbrs}
FCC.
\newblock {Title 47, Chapter I, Subchapter D, Part 96, Citizens Broadband Radio
  Service}.
\newblock Regulatory information, Federal Communications Commission, October
  2020.

\bibitem{ghosh2014torpath}
Mainak Ghosh, Miles Richardson, Bryan Ford, and Rob Jansen.
\newblock A torpath to torcoin: Proof-of-bandwidth altcoins for compensating
  relays.
\newblock Technical report, NAVAL RESEARCH LAB WASHINGTON DC, 2014.

\bibitem{Helium}
Amir Haleem, Andrew Allen, Andrew Thompson, Marc Nijdam, and Rahul Garg.
\newblock {Helium: A Decentralized Wireless Network}.
\newblock White paper, Helium Systems, Inc., November 2018.

\bibitem{harfoush2003measuring}
Khaled Harfoush, Azer Bestavros, and John Byers.
\newblock Measuring bottleneck bandwidth of targeted path segments.
\newblock In {\em IEEE INFOCOM 2003. Twenty-second Annual Joint Conference of
  the IEEE Computer and Communications Societies}, 2003.

\bibitem{hu2003evaluation}
Ningning Hu and Peter Steenkiste.
\newblock Evaluation and characterization of available bandwidth probing
  techniques.
\newblock {\em Journal on Selected Areas in Communications}, 21(6):879--894,
  2003.

\bibitem{jacobsonPathchar}
Van Jacobson.
\newblock Pathchar.
\newblock {\em ftp://ftp.ee.lbl.gov/pathchar/}.

\bibitem{jacobsonTraceroute}
Van Jacobson.
\newblock Traceroute.
\newblock {\em https://linux.die.net/man/8/traceroute6}.

\bibitem{jain2002end}
Manish Jain and Constantinos Dovrolis.
\newblock End-to-end available bandwidth: Measurement methodology, dynamics,
  and relation with tcp throughput.
\newblock In {\em ACM SIGCOMM Computer Communication Review}, 2002.

\bibitem{jain2002pathload}
Manish Jain and Constantinos Dovrolis.
\newblock Pathload: A measurement tool for end-to-end available bandwidth.
\newblock In {\em In Proceedings of Passive and Active Measurements (PAM)
  Workshop}. Citeseer, 2002.

\bibitem{josefsson2017edwards}
Simon Josefsson and Ilari Liusvaara.
\newblock Edwards-curve digital signature algorithm (eddsa).
\newblock Technical report, 2017.

\bibitem{karame2012security}
Ghassan~O Karame, Boris Danev, Cyrill Bannwart, and Srdjan Capkun.
\newblock On the security of end-to-end measurements based on packet-pair
  dispersions.
\newblock {\em IEEE Transactions on Information Forensics and Security},
  8(1):149--162, 2012.

\bibitem{keshav1991control}
Srinivasan Keshav.
\newblock A control-theoretic approach to flow control.
\newblock In {\em Proceedings of the conference on Communications architecture
  and protocols}, 1991.

\bibitem{kim2013improving}
Hyojoon Kim and Nick Feamster.
\newblock Improving network management with software defined networking.
\newblock {\em IEEE Communications Magazine}, 51(2):114--119, 2013.

\bibitem{kreutz2014software}
Diego Kreutz, Fernando~MV Ramos, Paulo~Esteves Verissimo, Christian~Esteve
  Rothenberg, Siamak Azodolmolky, and Steve Uhlig.
\newblock Software-defined networking: A comprehensive survey.
\newblock {\em Proceedings of the IEEE}, 103(1):14--76, 2014.

\bibitem{lai2000measuring}
Kevin Lai and Mary Baker.
\newblock Measuring link bandwidths using a deterministic model of packet
  delay.
\newblock In {\em Proceedings of the conference on applications, technologies,
  architectures, and protocols for computer communication}, pages 283--294,
  2000.

\bibitem{lai2001nettimer}
Kevin Lai and Mary Baker.
\newblock Nettimer: A tool for measuring bottleneck link bandwidth.
\newblock In {\em 3rd USENIX Symposium on Internet Technologies and Systems
  (USITS 01)}, 2001.

\bibitem{magma}
{Magma: A modern mobile core network solution}.
\newblock Magma Core Foundation.

\bibitem{mah2000pchar}
B.~A. Mah.
\newblock pchar: A tool for measuring internet path characteristics.
\newblock {\em http://www. employees. org/~ bmah/Software/pchar/}, 2000.

\bibitem{melander2000new}
Bob Melander, Mats Bjorkman, and Per Gunningberg.
\newblock A new end-to-end probing and analysis method for estimating bandwidth
  bottlenecks.
\newblock In {\em Globecom'00-IEEE. Global Telecommunications Conference}.
  IEEE, 2000.

\bibitem{melander2002regression}
Bob Melander, Mats Bjorkman, and Per Gunningberg.
\newblock Regression-based available bandwidth measurements.
\newblock In {\em International Symposium on Performance Evaluation of Computer
  and Telecommunications Systems}, 2002.

\bibitem{merkle1987digital}
Ralph~C Merkle.
\newblock A digital signature based on a conventional encryption function.
\newblock In {\em Conference on the theory and application of cryptographic
  techniques}, pages 369--378. Springer, 1987.

\bibitem{mills1989on}
David~L Mills.
\newblock On the accuracy and stablility of clocks synchronized by the network
  time protocol in the internet system.
\newblock In {\em ACM SIGCOMM Computer Communication Review}, 1989.

\bibitem{oran}
{ORAN: Transforming the Radio Access Networks Towards Open, Intelligent,
  Virtualized and Fully Interoperable RAN}.
\newblock O-RAN Alliance e.V.

\bibitem{pasztor2002active}
Attila Pasztor and Darryl Veitch.
\newblock Active probing using packet quartets.
\newblock In {\em Proceedings of the 2nd ACM SIGCOMM Workshop on Internet
  Measurment}, pages 293--305, 2002.

\bibitem{peng2003protection}
Tao Peng, Christopher Leckie, and Kotagiri Ramamohanarao.
\newblock Protection from distributed denial of service attacks using
  history-based ip filtering.
\newblock In {\em IEEE International Conference on Communications, 2003.
  ICC'03.}, volume~1, pages 482--486. IEEE, 2003.

\bibitem{perahia2013next}
Eldad Perahia and Robert Stacey.
\newblock {\em Next generation wireless LANs: 802.11 n and 802.11 ac}.
\newblock Cambridge university press, 2013.

\bibitem{prasad2003bandwidth}
Ravi Prasad, Constantine Dovrolis, Margaret Murray, and KC~Claffy.
\newblock Bandwidth estimation: metrics, measurement techniques, and tools.
\newblock {\em IEEE network}, 17(6):27--35, 2003.

\bibitem{ribeiro2003pathchirp}
Vinay~Joseph Ribeiro, Rudolf~H Riedi, Richard~G Baraniuk, Jiri Navratil, and
  Les Cottrell.
\newblock pathchirp: Efficient available bandwidth estimation for network
  paths.
\newblock In {\em Passive and active measurement workshop}, 2003.

\bibitem{salehin2013packet}
Khondaker~M Salehin and Roberto Rojas-Cessa.
\newblock Packet-pair sizing for controlling packet dispersion on wired
  heterogeneous networks.
\newblock In {\em 2013 International Conference on Computing, Networking and
  Communications (ICNC)}, pages 1031--1035. IEEE, 2013.

\bibitem{snader2009eigenspeed}
Robin Snader and Nikita Borisov.
\newblock Eigenspeed: secure peer-to-peer bandwidth evaluation.
\newblock In {\em IPTPS}, page~9, 2009.

\bibitem{strauss2003measurement}
Jacob Strauss, Dina Katabi, and Frans Kaashoek.
\newblock A measurement study of available bandwidth estimation tools.
\newblock In {\em Proceedings of the 3rd ACM SIGCOMM conference on Internet
  measurement}. ACM, 2003.

\bibitem{Althea}
Jehan Tremback and Justin Kilpatrick.
\newblock {Althea: An incentivized mesh network protocol}.
\newblock White paper, Althea Network, Inc., May 2017.

\bibitem{yang2022mobile}
Xinlei Yang, Hao Lin, Zhenhua Li, Feng Qian, Xingyao Li, Zhiming He, Xudong Wu,
  Xianlong Wang, Yunhao Liu, Zhi Liao, et~al.
\newblock Mobile access bandwidth in practice: measurement, analysis, and
  implications.
\newblock In {\em Proceedings of the ACM SIGCOMM 2022 Conference}, pages
  114--128, 2022.

\bibitem{yang2021fast}
Xinlei Yang, Xianlong Wang, Zhenhua Li, Yunhao Liu, Feng Qian, Liangyi Gong,
  Rui Miao, and Tianyin Xu.
\newblock Fast and light bandwidth testing for internet users.
\newblock In {\em 18th USENIX Symposium on Networked Systems Design and
  Implementation (NSDI 21)}, pages 1011--1026, 2021.

\bibitem{zhou2015magic}
Peng Zhou, Rocky~KC Chang, Xiaojing Gu, Minrui Fei, and Jianying Zhou.
\newblock Magic train: design of measurement methods against bandwidth
  inflation attacks.
\newblock {\em IEEE Transactions on Dependable and Secure Computing},
  15(1):98--111, 2015.

\end{thebibliography}

\end{document}